\documentclass[conference]{IEEEtran}
\IEEEoverridecommandlockouts
\usepackage{graphicx}
\usepackage{enumitem}
\usepackage{xcolor}
\usepackage{cite}
\usepackage{amsmath,amssymb,amsfonts}
\usepackage{algorithmic}
\usepackage{textcomp}
\usepackage{calc}
\usepackage{url} 
\usepackage{makecell}

\usepackage{tcolorbox}
\usepackage{booktabs}   
\usepackage{ltablex}
\usepackage{amsthm}
\usepackage[para]{footmisc}
\usepackage{subcaption}
\usepackage{xurl}
\usepackage{hyperref}
\hypersetup{colorlinks = true, allcolors = black}

\newcounter{researchquestion}
\newtcolorbox{researchquestionbox}[1][]{%
  colback=white, 
  colframe=black!75!black,
  fonttitle=\bfseries,
  title=Challenge \arabic{researchquestion}:,
  before upper={\stepcounter{researchquestion}},
  #1
}
\def\BibTeX{{\rm B\kern-.05em{\sc i\kern-.025em b}\kern-.08em
    T\kern-.1667em\lower.7ex\hbox{E}\kern-.125emX}}
\newtheorem{theorem}{Theorem}

\usepackage[linesnumbered,ruled,vlined]{algorithm2e}

\newtheorem{definition}{Definition}[section]

\def\BibTeX{{\rm B\kern-.05em{\sc i\kern-.025em b}\kern-.08em
    T\kern-.1667em\lower.7ex\hbox{E}\kern-.125emX}}
\begin{document}

\title{How Does Stake Distribution Influence Consensus? Analyzing Blockchain Decentralization }
\author{\IEEEauthorblockN{Shashank Motepalli,
Hans-Arno Jacobsen}
\IEEEauthorblockA{Department of Electrical and Computer Engineering,
University of Toronto\\
shashank.motepalli@mail.utoronto.ca,
jacobsen@eecg.toronto.edu}}


\maketitle
\vspace{-16pt}
\begin{abstract}
In the PoS blockchain landscape, the challenge of achieving full decentralization is often hindered by a disproportionate concentration of staked tokens among a few validators. This study analyses this challenge by first formalizing decentralization metrics for weighted consensus mechanisms. An empirical analysis across ten permissionless blockchains uncovers significant weight concentration among validators, underscoring the need for an equitable approach. To counter this, we introduce the Square Root Stake Weight (SRSW) model, which effectively recalibrates staking weight distribution. Our examination of the SRSW model demonstrates notable improvements in the decentralization metrics: the Gini index improves by 37.16\% on average, while Nakamoto coefficients for liveness and safety see mean enhancements of 101.04\% and 80.09\%, respectively. This research is a pivotal step toward a more fair and equitable distribution of staking weight, advancing the decentralization in blockchain consensus mechanisms.
\end{abstract}

\section{Introduction}
Bitcoin shaped the field of blockchains by introducing a peer-to-peer system that operates without trusted intermediaries~\cite{nakamoto2008bitcoin}. Its vision encapsulates \textit{decentralization}, characterized by the elimination of single points of failure and the facilitation of collective decision-making. Our work focuses on exploring decentralization in blockchains, particularly in their consensus mechanisms. These mechanisms are critical as they establish agreement on the content and order of transactions among validators.

The problem this paper addresses is the \textit{analysis and advancement of decentralization in consensus mechanisms}. This problem is particularly interesting because, although decentralization is fundamental to every blockchain, standardized metrics to quantify it within consensus are lacking. This gap, coupled with the technical complexity of consensus algorithms, complicates the analysis of blockchain systems in practice. Moreover, enhancing decentralization in consensus mechanisms is of significance, as it directly contributes to the trust and safety in blockchain systems.

Our approach to examining decentralization begins with a systematic classification of consensus mechanisms based on finality, as detailed in Section~\ref{sec:consensusclassification}. In this paper, we specifically focus on classical consensus mechanisms, leveraging their extensive research in distributed computing~\cite{lamport2019byzantine,castro1999practical}. Subsequently, we delve into the concept of weighted consensus wherein validators could have unequal influence in consensus, a framework highly relevant in Proof of Stake (PoS) and Delegated Proof of Stake (DPoS) systems. In these systems, the weighting is a function of the tokens staked by validators (or stakers)~\cite{david2018ouroboros,buchman2016tendermint,saad2020comparative}. 

Drawing inspiration from quantitative decentralization metrics in decentralized autonomous organizations (DAOs) with token-based voting~\cite{austgen2023dao,sharma2023unpacking,tan2023open}, our work extends these studies to consensus mechanisms in Section~\ref{sec:metrics}. We adapt and refine these metrics to effectively measure decentralization in consensus mechanisms. Specifically, we evaluate  cardinality, Gini index, and Nakamoto coefficients for safety and liveness for given set of validators.

Utilizing metrics specifically designed for evaluating decentralization in consensus mechanisms, our research carries out an empirical analysis of ten prominent blockchains, as elaborated in Section~\ref{sec:empiricalanalysis}. This analysis encompasses Aptos~\cite{aptos}, Axelar~\cite{axelar}, BNB~\cite{bnb}, Celestia~\cite{celestia}, Celo~\cite{celo}, Cosmos~\cite{cosmos}, Injective~\cite{injective}, Osmosis~\cite{osmosis}, Polygon~\cite{polygon}, and Sui~\cite{sui}. The findings reveal a notable concentration of weight among few validators, which poses significant concerns for the security and integrity of blockchain systems. To address this challenge, we propose \textit{the square root stake weight (SRSW)} function, outlined in Section~\ref{sec:srsw-model}. Unlike solutions such as the introduction of virtual stake~\cite{mivsic2023towards}, our approach leverages existing staked tokens while redefining their weights, thereby avoiding the introduction of new security vulnerabilities. Our proposed approach, thoroughly evaluated in Section~\ref{sec:evaluation}, demonstrates substantial potential in enhancing decentralization and, consequently, the overall reliability of blockchain systems. The paper concludes with a discussion of related work in Section~\ref{sec:relatedwork}, and identifies avenues for future research in Section~\ref{sec:conclusion}.

Our contributions are four-fold: 
\vspace{-2pt}
\begin{enumerate}
    \item Adapt and formalize decentralization metrics specifically for consensus mechanisms.
    \item We empirically demonstrate the challenge of weight concentration and its impact on decentralization in prominent blockchains. 
    \item We introduce the square root stake weight (SRSW) mechanism, to mitigate weight concentration challenge. 
    \item Our work pioneers the use of data-driven analysis in consensus research, offering novel insights and solutions to enhance decentralization in blockchain systems.
\end{enumerate}
\vspace{-5pt}
\section{Consensus Mechanism Foundations and Classification}
\label{sec:consensusclassification}
This section delineates the consensus mechanisms and validator set selection process, setting the stage for the empirical analysis of decentralization in blockchains.
\subsection{Classification of Consensus Mechanisms Based on Finality}
A blockchain comprises a ledger that has blocks of transactions. The consensus mechanism orchestrates the process of reaching agreement on the content and order of blocks within the ledger~\cite{zhang2022reaching}. Agreement is reached among a designated set of validators, which in different nomenclatures can be referred to as miners~\cite{nakamoto2008bitcoin}, witnesses~\cite{li2020comparison}, or sequencers~\cite{motepalli2023sok}. A consensus mechanism is deemed Byzantine Fault Tolerant (BFT) if it withstands a certain proportion of validators with malicious behaviour, in addition to crash failures~\cite{lamport2019byzantine}.

Two properties are guaranteed by consensus mechanisms: \textit{safety}, ensuring all correct validators agree on the same content and order of blocks, and \textit{liveness}, ensuring the continual production of new blocks without indefinite delays~\cite{pass2017analysis,garay2015bitcoin}. Building upon the safety property, we introduce the concept of finality, also known as commitment. Finality of a block \( b \) at time $t$, denoted \( 0 \leq f(b, t) \leq 1\), indicates the probability with which the block has been appended to the ledger~\cite{anceaume2020finality}. When \( f(b, t) =1 \), it signifies \textit{total finality}, i.e., the block $b$ cannot be reverted or abandoned. Achieving this level of finality is essential for the immutability of the ledger. There are two distinct ways to realize finality in consensus mechanisms:
\vspace{-3pt}
\begin{definition}
\label{def:absolute-finality}
\textbf{Absolute Finality} is achieved when a block \( b \) is appended to the ledger at time $t_0$ and becomes irreversible instantly, such that \( f(b, t) = 1\quad\text{for all } t > t_0 \).
\end{definition}
\vspace{-5pt}

\begin{definition}
\label{def:probabilistic-finality}
\textbf{Probabilistic Finality} is achieved when the finality of an appended block \( b \) at time \( t_0 \) is expressed as \( f(b, t_0) = 1-\gamma \), where \( \gamma \) is a non-negative value less than 1 (i.e., \( 0 \leq \gamma < 1\)) that represents the deviation from total finality. Specifically, for any two points in time \( t_1 \) and \( t_2 \) such that \( t_2 > t_1 \), it follows that \( f(b, t_1) < f(b, t_2) \). As time progresses, \( f(b, t) \) gradually converges to 1 as \( \gamma \) approaches zero. 
\begin{equation}
\lim_{\gamma \to 0} f(b, t) = 1
\vspace{-6pt}
\end{equation}
\end{definition}

Based on how consensus mechanisms achieve finality, they can be categorized into two types: Nakamoto-style and classical consensus, as shown in Table~\ref{tab:types-consensus-mechanisms}.
\begin{table}[]
\caption{Consensus mechanisms classified based on finality}
\label{tab:types-consensus-mechanisms}
\renewcommand{\arraystretch}{1.5}
\centering
\begin{tabular}{|l|l|l|}
\hline
\multicolumn{1}{|c|}{} & \multicolumn{1}{c|}{\textbf{Classical consensus}}                     & \multicolumn{1}{c|}{\textbf{Nakamoto-style}}                          \\ \hline \hline
\textbf{Finality}      & \begin{tabular}[c]{@{}l@{}}Absolute and \\ instant\end{tabular} & \begin{tabular}[c]{@{}l@{}}Probabilistic and \\ eventual\end{tabular} \\ \hline
\textbf{Priniciple}    & Safety over liveness                                                  & Liveness over safety                                                  \\ \hline
\textbf{Attestation}   & Quorum of validators                                                  & Only the proposer                                                     \\ \hline
\textbf{Resources}     & Requires a priori knowledge                                            & No constraints                                                        \\ \hline
\textbf{Communication} & Supports partial synchrony                                            & Synchronous                                                           \\ \hline
\textbf{Examples}      & \begin{tabular}[c]{@{}l@{}}PBFT~\cite{castro1999practical}, HotStuff~\cite{yin2019hotstuff}, \\ Tendermint~\cite{buchman2018latest,buchman2016tendermint}\end{tabular} & \begin{tabular}[c]{@{}l@{}}Nakamoto~\cite{nakamoto2008bitcoin},  \\Ouroborus~\cite{david2018ouroboros}\end{tabular}        \\ \hline
\textbf{In practise}   & Diem~\cite{baudet2019state}                                                                  & Bitcoin~\cite{nakamoto2008bitcoin}                                                               \\ \hline
\end{tabular}
\vspace{-16pt}
\end{table}

The Nakamoto-style consensus embodies probabilistic finality, meaning that the system eventually approaches total finality with time~\cite{kim2023taxonomic}. This style is used in Bitcoin, where a block is considered to have reached finality after the confirmation of 6 subsequent blocks, approximately an hour~\cite{nakamoto2008bitcoin}. Nakamoto-style consensus mechanisms prioritize liveness over safety, ensuring the continual production of new blocks; however, the ledger's order remains susceptible to forking~\cite{garay2015bitcoin},  i.e., the current order of the ledger may be altered, until time \( t \).

Conversely, classical consensus mechanisms prioritize safety over liveness. In this style, no blocks are appended to the ledger until absolute finality is achieved, rendering finality deterministic and immediate~\cite{kim2023taxonomic}. An example is the PBFT consensus~\cite{castro1999practical}, where a designated proposer, one of the validators, broadcasts a block and, absolute finality is achieved when a quorum of validators attests on the proposed block and the block is appended to the ledger.

A notable distinction between these styles also lies in the block generation process. In Nakamoto-style consensus, a single proposer is responsible for proposing a block. This design doesn't presuppose knowledge of resources, such as the hash power in PoW, and assumes synchronous communication, i.e., messages are broadcast within a bounded time~\cite{lewis2021does}. On the other hand, classical consensus protocols assume a priori knowledge of the total available resources~\cite{buchman2016tendermint}, such as the stake distribution of validators in PoS. The notion of a quorum $\mathbb{Q}$, facilitated through certificates of attestation, necessitates having finite resources, a topic explored further in the subsequent subsection.

For the scope of this work, our focus is classical consensus mechanisms for multiple reasons. Primarily, these mechanisms facilitate fast finality along with high performance in terms of throughput and latency, compared to Nakamoto-style consensus~\cite{yin2019hotstuff,buchman2016tendermint}. Secondly, the employment of quorum $\mathbb{Q}$ certificates enables these protocols to function effectively in a partially synchronous environment, thereby tolerating indefinite periods of asynchrony~\cite{lewis2021does}. Thirdly, classical consensus mechanisms have undergone rigorous examination over several decades~\cite{zhang2022reaching}, with seminal contributions such as Raft~\cite{ongaro2014search} and PBFT~\cite{castro1999practical}, finding applications in safety-critical domains such as aviation systems~\cite{wensley1978sift,siewiorek2005fault}.

In the subsequent sections, the discussion extends to Sybil resistance and the intricacies of weighted classical consensus.

\subsection{Transition to Weighted Consensus}
Traditional classical consensus mechanisms, such as PBFT~\cite{castro1999practical}, HotStuff~\cite{yin2019hotstuff}, PrestigeBFT~\cite{zhang2023prestigebft}, and SBFT~\cite{gueta2019sbft}, are designed to be able to tolerate up to one-third of the validator set being faulty, where a faulty validator may exhibit malicious behavior or be offline. In these mechanisms, a designated block proposer is required to collect attestations from the validator set to form a quorum certificate. Let the validator set be represented as \( N = \{n_1, n_2, \ldots, n_m \} \) where $n_i$ represents a validator. A quorum certificate is formed with attestations from at least a super-majority of validators, denoted as \( \mathbb{Q} \), such that:
\vspace{-2pt}
\begin{equation}
\mathbb{Q} \geq \left(\frac{2}{3}\right) m \text{, where } m = |N|
\vspace{-2pt}
\end{equation}
Note that while we assume the protocol can be able to tolerate up to one-third of the total validators being faulty, some protocols may have different failure assumptions~\cite{miller2016honey}. In such cases, $\mathbb{Q}$ must be adjusted accordingly.

Classical BFT consensus mechanisms were initially conceived for permissioned systems, where the identities of all validators are established. When deployed in permissionless environments with (pseudo)anonymous identities, these mechanisms become susceptible to Sybil attacks, wherein a malicious actor could create multiple validators to subvert the consensus process~\cite{douceur2002sybil}. To mitigate this vulnerability and achieve Sybil resistance without relying on trusted intermediaries, Algorand~\cite{gilad2017algorand}, Ouroboros~\cite{david2018ouroboros} and Tendermint~\cite{buchman2016tendermint} pioneered PoS mechanism. In PoS, validators stake the native tokens of the system as a means of establishing their identity. These tokens are subject to penalization if validators engage in malicious behavior~\cite{motepalli2021reward}. Given that the native tokens are finite and the security of the consensus impacts the tokens' market value, validators are rationally incentivized to act correctly, thus enhancing the system's security. The development of PoS protocols, characterized by variably staked tokens, paves the way for the adoption of weighted consensus.

\textit{Weighted consensus} encompasses traditional classical consensus as a subset, where traditional models are effectively a special instance with uniform weights across validators. In weighted consensus, validators have varying weights in the consensus process~\cite{edwardThesis}. In practice, in PoS/DPoS blockchains such as Cosmos~\cite{buchman2016tendermint}, the influence of a validator \( n_k \) in the consensus is quantified by their weight, \( w_k > 0\), which is a function of their staked tokens \( s_k\).
\vspace{-2pt}
\begin{equation}
    w_k = s_k
\vspace{-2pt}
\end{equation}
Unlike traditional classical consensus mechanisms where a quorum ceritifcate is achieved based on the absolute number of validators, in weighted consensus, a quorum necessitates garnering two-thirds of the total weight. The quorum certificate for weighted consensus for a validator set $N$ is denoted as \( \mathbb{Q}' \), is given below:
\vspace{-2pt}
\begin{equation}
\mathbb{Q}' \geq \left( \frac{2}{3}\right) \sum_{k} w_k \quad \forall k \in N
\vspace{-2pt}
\end{equation}
 Moreover, higher weight could also imply higher rewards or a higher probability of being selected as a block proposer~\cite{david2018ouroboros}. The evolution from traditional to weighted consensus, underscored by PoS/DPoS, is a pivotal adaptation to suit permissionless blockchain systems. We focus on weighted consensus in the rest of this work.

\subsection{Validator Set Selection}
Classical weighted consensus assumes that resources, such as total staked tokens in PoS, are finite and known a priori. These staked tokens are used to rank candidates interested in becoming validators and to choose the validator set. Various mechanisms exist for validator set selection, including PoS~\cite{gilad2017algorand}, DPoS~\cite{saad2020comparative}, delay towers~\cite{motepalli2022decentralizing}, and reputation mechanisms~\cite{de2018pbft}. This work does not make specific assumptions regarding the mechanism of validator set selection; instead, it focuses on the validators' engagement in the consensus mechanism.

The validator set typically remains fixed for a specified time interval, known as an \textit{epoch}. Following each epoch, a new validator set is selected through a \textit{reconfiguration} event~\cite{duan2022foundations,motepalli2022decentralizing}. Reconfiguration tends to consider updated stakes and involves eliminating faulty validators. 

It is also important to acknowledge that some blockchains, such as Algorand~\cite{gilad2017algorand}, use mechanisms like random sortation to randomly select a subset of candidates as validators every epoch. Our study focuses on systems where the validator set is deterministically defined, thereby excluding blockchains that employ random committee selection processes.

\section{Consensus Decentralization Metrics}
\label{sec:metrics}
In consensus mechanisms, \textit{decentralization} means reaching agreement on the contents and order of transactions without centralized control, ensuring that no single validator or group of validators dominates the process. While challenging to precisely define~\cite{schneider2003decentralization,sharma2023unpacking,kiayias2022sok}, decentralization is essential for consensus mechanisms, as it underpins trust in the blockchain systems.

Our discussion draws inspiration from the \( (m,\varepsilon,\delta) \)-decentralization model described in ``Impossibility of Full Decentralization in Permissionless Blockchains''~\cite{kwon2019impossibility}. Here, \( m \) indicates the cardinality of the validator set, 
and \( \varepsilon \) represents the weight disparity between the most influential (richest) 
and the \( \delta \)-th percentile validator. The ideal case is full decentralization, expressed as \( (m,0,0) \) for a sufficiently large \( m \), that occurs when all validators have equal influence. While the \( (m,\varepsilon,\delta) \)-model captures the essence of decentralization, 
it lacks quantifiable metrics for comparing the decentralization of different blockchains. Therefore, we introduce additional metrics, as shown in Table~\ref{tab:metrics_symbols}, to effectively quantify and compare decentralization across different blockchains.

\subsubsection{\textbf{Validator Set Cardinality} (\( m \))}
\begin{description}[leftmargin=0em, labelwidth=0em, align=left, font=\bfseries]
    \item[Description:] Represents the number of validators (\( m = |N| \)) in the consensus mechanism.
    \item[Inference:] A higher \( m \) suggests better decentralization, aligning with the \( (m,\varepsilon,\delta) \)-model.
    \item[Limitations:] In weighted consensus, \( m \) alone may not reflect true decentralization. For example, if \( m=1000 \) but one validator holds 90\% of the total weight, it contradicts the decentralization ideal.
\end{description}

\subsubsection{\textbf{Gini Coefficient} (\( G \))}
\begin{description}[leftmargin=0em, labelwidth=0em, align=left]
    \item[Description:] The Gini coefficient (\( G \)) measures wealth inequality, commonly used in socioeconomic studies~\cite{ceriani2012origins,gini1921measurement,sitthiyot2020simple}. In consensus mechanisms, it assesses validators' influence disparity, indicating deviation from \( (m,0,0) \)-decentralization.

   \( G \) is calculated using the Lorenz curve, which graphically elucidates the weight distribution among validators~\cite{gastwirth1972estimation}. The Lorenz curve plots the cumulative share of validators (sorted by their weight) on the X-axis against the cumulative share of their weight on the Y-axis. The formula for \( G \) is:
   \vspace{-8pt}
    \begin{equation}
        G = 1 - \frac{2 \times B}{A + B}
    \end{equation}
    where \( B \) is the area between the Line of Equality and the Lorenz Curve, and \( A \) is the area beneath the Lorenz Curve. The line of equality illustrates a hypothetical scenario of equal weight distribution, while the Lorenz curve depicts the actual distribution of weights~\cite{sitthiyot2021simple}. The area between these two curves represents the extent of inequality in the weight distribution~\cite{gastwirth1972estimation}.
        
    \item[Inference:] \( G \) ranges between 0 and 1, with 0 indicating equitable distribution (higher decentralization) and values closer to 1 indicating concentration of weight (lower decentralization).
    \item[Limitations:] \( G \) alone may not fully capture decentralization, as it does not account for validator set cardinality \( m \). For example, a system with a single validator (\( m =1\)) would have \( G =0 \), yet be highly centralized.
\end{description}
\begin{table}[]
\renewcommand{\arraystretch}{1.5}
\centering
\caption{Decentralization metrics for consensus }
\label{tab:metrics_symbols}
\begin{tabular}{|l|l|l|l|}
\hline
\textbf{Symbol} & \textbf{Metric}                                                          & \textbf{Range}                                                                            & \textbf{Ideal} \\ \hline \hline
\( m \)               & Validator set cardinality                                                & \( m > 0\)                                                                          & higher        \\ \hline
\( G \)               & Gini Index                                                               &  \( 0 \leq G \leq 1\)                                                                & lower          \\ \hline
\( \mathbb{N}_L, \rho_{\mathbb{N}_L} \)             & \begin{tabular}[c]{@{}l@{}}Nakamoto Coefficient \\ for Liveness\end{tabular} & \begin{tabular}[c]{@{}l@{}}\( \mathbb{N}_L \geq 0 \)\\ \( 0\leq \rho_{\mathbb{N}_L} \leq 100\)\end{tabular} & higher         \\ \hline
\( \mathbb{N}_S, \rho_{\mathbb{N}_S} \)               & \begin{tabular}[c]{@{}l@{}}Nakamoto Coefficient\\ for Safety\end{tabular}    & \begin{tabular}[c]{@{}l@{}}\( \mathbb{N}_S \geq 0 \)\\ \( 0\leq \rho_{\mathbb{N}_S} \leq 100\)\end{tabular} & higher        \\ \hline
\end{tabular}
\vspace{-16pt}
\end{table}

\subsubsection{\textbf{Nakamoto Coefficient - Liveness} (\( \mathbb{N}_L, \rho_{\mathbb{N}_L} \))}
\begin{description}[leftmargin=0em, labelwidth=0em, align=left]
    \item[Description:] The Nakamoto coefficient for liveness (\( \mathbb{N}_L, \rho_{\mathbb{N}_L} \)) quantifies the minimum number of validators needed to disrupt the block production or censor transactions~\cite{ofacSanctioned,balajidecentralization,censorshipData}. In other words, \( \mathbb{N}_L \) indicates the fault tolerance in weighted consensus, i.e., cardinality of the smallest subset of validators ($L$) whose cummulative weight is at least one-third of the total weight. The formula is:
    \begin{equation}
    \mathbb{N}_L = \min\{|L| \,|\, L \subseteq N, \sum_{i \in L} w_i \geq \frac{1}{3} \sum_{i=1}^{m} w_i\}
    \end{equation}
    The scaled Nakamoto coefficient \( \rho_{\mathbb{N}_L} \) is then calculated as:
    \begin{equation}
    \rho_{\mathbb{N}_L} = \frac{\mathbb{N}_L}{m} * 100
    \end{equation}
    This normalization facilitates comparison across blockchains of varying sizes. Unlike the \( (m,\varepsilon,\delta) \)-decentralization model, which focuses on the richest validator, \( \rho_{\mathbb{N}_L} \) considers the weight distribution across the top one-third of validators relative to the entire set, offering a broader view of the system's decentralization.
    
    \item[Inference:] A higher \( \mathbb{N}_L \) suggests better decentralization. The normalized value \( \rho_{\mathbb{N}_L} \), ranging between 0 and 100. A \( \rho_{\mathbb{N}_L} \) closer to 100 indicates high decentralization and resilience against censorship~\cite{ofacSanctioned}.
    \item[Limitations:]  \( \mathbb{N}_L \) relies on correct majority of validators. Any collusion among validators may distort \( \mathbb{N}_L \), rendering it an inaccurate measure of decentralization. 
\end{description}

\subsubsection{\textbf{Nakamoto Coefficient - Safety} (\( \mathbb{N}_S, \rho_{\mathbb{N}_S} \))}
\begin{description}[leftmargin=0em, labelwidth=0em, align=left]

    \item[Description:] Safety relies on finality—ensuring once blocks are appended, they become immutable. Safety compromises have severe consequences on the ledger's integrity, such as ledger re-ordering or loss of funds~\cite{li2023hard,sridhar2023better}. The Nakamoto Coefficient for Safety (\( \mathbb{N}_S, \rho_{\mathbb{N}_S} \)) quantifies the minimum subset of the validators (\( S \)) required to compromise safety~\cite{balajidecentralization,pass2017analysis}, i.e., \( \mathbb{N}_S \) is the cardinality of the smallest subset of validator set whose combined weight can form a quorum \( \mathbb{Q}' \):
    \begin{equation}
    \mathbb{N}_S = \min\{|S| \,|\, S \subseteq N, \sum_{i \in S} w_i \geq \mathbb{Q}' \}
    \end{equation}
    The scaled form of \( \mathbb{N}_S \), denoted by \( \rho_{\mathbb{N}_S} \), is:
    \begin{equation}
    \rho_{\mathbb{N}_S} = \frac{\mathbb{N}_S}{m}
    \end{equation}
    In conjunction with (\( \mathbb{N}_L, \rho_{\mathbb{N}_L} \)), the (\( \mathbb{N}_S, \rho_{\mathbb{N}_S} \)) metric complements the \( (m,\varepsilon,\delta) \)-decentralization framework by quantifying the concentration of weight that affects system safety.
    \item[Inference:] Higher values of \( \mathbb{N}_S \) and \( \rho_{\mathbb{N}_S} \) indicate a more decentralized system. Specifically, \( \rho_{\mathbb{N}_S} \) close to 100, especially in systems with a large validator set (\( m \)), signifies robust decentralization.

    \item[Limitations:] The assumption that validators in \( \mathbb{N}_S \) calculations are non-colluding may not reflect real-world scenarios, potentially limiting its accuracy for safety evaluation. Moreover, reducing a system's safety to a single metric like \( \mathbb{N}_S \) risks oversimplifying safety complexities~\cite{wang2021ethereum}.

\end{description}

\subsubsection{Summary}
In this section, we introduced various metrics to quantify decentralization in consensus mechanisms. Despite the limitations acknowledged, the synergy of these metrics holistically captures the essence of decentralization, aligning with the \( (m,\varepsilon,\delta) \)-decentralization model. We leverage these metrics to quantify the decentralization of existing blockchain systems in practice in the subsequent section.

\section{Empirical Data Analysis}
\label{sec:empiricalanalysis}
\begin{table*}[htp]
\vspace{-12pt}
\caption{Decentralization metrics for blockchains, as of 14 December 2023. Ideal values of G and \( \varepsilon\) are 0.}
\centering
\label{tab:metrics}
\renewcommand{\arraystretch}{1.5}
\begin{tabular}{|l|l|l|l|l|l|l|l|l|}
\hline
\multicolumn{1}{|c|}{} & \textbf{Application} & \textbf{Consensus Mechanism} & \textbf{\( m \)} & \( G \) & \( \rho_{\mathbb{N}_L} (\mathbb{N}_L) \) & \( \rho_{\mathbb{N}_S} (\mathbb{N}_S) \) & \textbf{\( \varepsilon \) in \( (m,\varepsilon,0) \)} & \( (m,\varepsilon,50) \) \\ \hline \hline
\textbf{Aptos}         & L1 blockchain~\cite{aptos}        & Jolteon/DiemBFT~\cite{aptosBlockchain,gelashvili2022jolteon} & 144        & 0.56       & 12.50 (18)        & 12.50 (38)        & 8.488454e+11 & 7.63       \\ \hline
\textbf{Axelar}        & Interoperability~\cite{axelar}     & Tendermint~\cite{axelarBlockchain}         & 75         & 0.41       & 13.33 (10)        & 37.33 (28)        & 7.796480e+03 & 5.01        \\ \hline
\textbf{BNB (Binance)} & L1 blockchain~\cite{bnb}        & Tendermint~\cite{bnbBlockchain}         & 57         & 0.55       & 14.04 (8)         & 28.07 (16)        & 1.595114e+05 & 8.41        \\ \hline
\textbf{Celestia} & Data availability~\cite{celestia}        & Tendermint~\cite{celestiaDocs}         & 174         & 0.83       & 2.87 (5)         & 8.62 (15)        & 3.836768e+10 & 88.86        \\ \hline
\textbf{Celo}          & L2* (L1 blockchain)~\cite{celo}                & IstanbulBFT~\cite{celoBlockchain}        & 84         & 0.40       & 11.90 (10)        & 39.29 (33)        & 1.293101e+10 & 3.90        \\ \hline
\textbf{Cosmos}        & L1/interoperability~\cite{cosmos}  & Tendermint~\cite{cosmosBlockchain}         & 180& 0.69       & 3.89 (7)          & 13.33 (24)        & 2.470500e+02 & 60.63       \\ \hline
\textbf{Injective} &  DeFi/interoperability~\cite{injective} & Tendermint~\cite{injectiveBlog}         & 60        & 0.49       & 8.33 (5)          & 30.00 (18)        & 3.158000e+01 & 8.08       \\ \hline
\textbf{Osmosis}       & DeFi/DEX~\cite{osmosis}             & Tendermint~\cite{osmosisBlockchain}         & 150        & 0.54       & 6.67 (10)         & 28 (42)        & 1.080500e+02 & 14.52      \\ \hline
\textbf{Polygon}       & L2/ZK-rollup~\cite{polygon}                   &Tendermint~\cite{polygonBlockchain}         & 105        & 0.78       & 3.81 (4)          & 10.48 (11)        & 3.629552e+08 & 69.53       \\ \hline
\textbf{Sui}           & L1 blockchain~\cite{sui}        & Narwhal/BullShark~\cite{suiBlockchain}  & 106        & 0.41       & 13.21 (14)        & 33.02 (35)        & 9.290000e+00 & 6.37        \\ \hline
\end{tabular}
\vspace{-16pt}
\end{table*}

\subsection{Scope and Methodology of Data Collection}
In our study, we focus on permissionless blockchains that use classical consensus mechanisms, particularly those with weighted consensus as outlined in Section~\ref{sec:consensusclassification}. We limit our analysis to blockchains with deterministic validator set selection methods, such as PoS and DPoS. Notably, all the blockchains examined in this study employ DPoS for Sybil resistance.

For our empirical analysis, we selected ten blockchains, as shown in Table~\ref{tab:metrics}, namely: Aptos, Axelar, BNB (Binance), Celestia, Celo, Cosmos, Injective, Osmosis, Polygon, and Sui. These protocols are based on BullShark~\cite{spiegelman2022bullshark}, HotStuff~\cite{yin2019hotstuff}, IstanbulBFT~\cite{moniz2020istanbul}, and Tendermint/CometBFT~\cite{buchman2016tendermint} consensus mechanisms. These blockchains were chosen for their diverse applications, including smart contracts Layer-1 (L1) blockchains, interoperability protocols~\cite{belchior2021survey}, Layer-2 (L2) scaling solutions~\cite{mccorry2021sok}, data availability protocols~\cite{al2019lazyledger}, and decentralized exchanges~\cite{lehar2021decentralized}, ensuring the wide applicability of our findings. Our sample comprises blockchains with at least 300 million USD market capitalization, cumulatively amounting to a market capitalization of 60 billion USD as of December 14, 2023~\cite{coinmarketcap}.

Data collection was automated via scripts interfacing with the blockchains' RPC endpoints to fetch active validator sets and their staked tokens. Daily snapshots were taken to account for epoch changes and the data was systematically archived in a public GitHub repository\footnote{https://github.com/sm86/destake} to facilitate transparency and accessibility for ongoing and subsequent analyses.
\subsection{Data Analysis}

In Table~\ref{tab:metrics}, we present the decentralization metrics for ten blockchains, derived from the validator set data snapshot on 14 December 2023. To validate the reliability of our analysis, we continuously monitored over the prior month, confirming the absence of significant deviations in the observed trends.

Upon examining the validator set cardinality (\( m \)),  we note a range of 57 to 180 validators across the analysed blockchains. However, a concerning trend emerges when we consider the Nakamoto coefficients for liveness (\( \mathbb{N}_L\)) and safety (\( \mathbb{N}_S \)), which are notably low, varying from 4 to 18 and 11 to 42, respectively. This disparity implies potential vulnerabilities. For example, in the case of Polygon, merely the top 4 validators could censor an application~\cite{censorshipData}, furthermore, the top 11 validators might collude to alter the ledger. Despite a high number of validators, the proportion that could compromise system liveness (\( \rho_{\mathbb{N}_L} \)) and safety (\( \rho_{\mathbb{N}_S} \)) remains worryingly small, spanning only 2.87\% to 14.04\% and 8.62\% to 39.29\%, respectively. This observation leads us to an open challenge to bolster both the liveness and safety of blockchain systems, utilizing the available validators.
\vspace{-6pt}
\begin{researchquestionbox}
Given a validator set with cardinality (\( m \)), how can we enhance the Nakamoto coefficients (\( \rho_{\mathbb{N}_L}, \rho_{\mathbb{N}_S} \))?
\end{researchquestionbox}
\vspace{-5pt}
An examination of the Gini coefficient (\( G \)) in Table~\ref{tab:metrics}, ranging from 0.40 to 0.83, reveals a significant concentration of weight. This is indicative of a disproportionate stake distribution among a small subset of validators. The implications of this concentration are further underscored by the values of $\varepsilon$ in the $(m, \varepsilon, \delta)$-decentralization model. Both at $\delta = 0$ and $\delta = 50$, the $\varepsilon$ values substantially deviate from the ideal zero value needed for $(m, 0, 0)$, i.e., full decentralization. This observation leads us to our second critical challenge in the pursuit of enhancing blockchain decentralization. 
\vspace{-6pt}
\begin{researchquestionbox}
    How can we achieve a more equitable weight distribution (\( G \)) among validators?
\end{researchquestionbox}
\vspace{-5pt}
A simplistic approach to addressing these challenges might be to adopt a one-validator one-vote system with equal weight for all validators, complemented by a minimum stake threshold for participation. However, this approach shows limitations in practice in blockchains such as Ethereum. With the requirement of only 32 ETH for validator eligibility, Ethereum's network consists of over 800,000+ validators~\cite{ethereumCoindesk}. This high number necessitates random selection for consensus participation, impacting performance metrics such as time to finality. A more significant concern is operational centralization, exemplified by Lido controlling approximately 32.7\% of all validators~\cite{lidoArticle}, effectively making the Nakamoto coefficient for liveness valued at one (\(\mathbb{N}_L=1\)). Therefore, the problem lies in mitigating the risk of Sybil identities to achieve genuine decentralization, captured in the folowing challenge.
\vspace{-6pt}
\begin{researchquestionbox}
How to discourage creation of multiple Sybil identities? 
\end{researchquestionbox}
\vspace{-10pt}
\section{Advancing Decentralization: Finite SRSW Quorums}
\label{sec:srsw-model}
In this section, we address the identified challenges. We begin by defining primitives needed for our proposed solution. This is followed by the introduction of the SRSW (Square Root Stake Weighted) quorum. 
\subsection{Primitives and Assumptions}

\textbf{Validator Rewards:} 
Validators are incentivized with rewards for their participation in the consensus mechanism. We assume validators are rational and want to maximize their rewards. 

At the end of every epoch, each correct validator \( n_k \in N \) receives a reward, denoted as \( r_{n_k} \). The reward is calculated based on the system parameter \( \alpha \), representing the inflation factor that determines the rate of reward distribution. The reward for each validator is given by the equation:
\begin{equation}
    r_{n_k} = \alpha w_k, \quad \text{where} \quad \alpha > 0.
\end{equation}
If a correct validator holds \( s_k \) native tokens at the start of an epoch, their balance of native tokens after that epoch would increase to \( s_k + r_{n_k} \).

\textbf{Sybil Cost:}
We introduce a Sybil Cost, denoted as \( C > 0 \), to represent the additional expenses incurred by a validator operator when choosing to run multiple validator nodes instead of one. These expenses could include operational costs, such as computational resources or the amount of staked tokens required. We assume that \( C \) is sufficiently high, thereby providing resistance against Sybil attacks.

\textbf{Limit Validator Set Cardinality:}
We propose an upper limit for the validator set cardinality  \( M \), ensuring \( m \leq M \). To implement this in practice, the validator candidates are sorted based on their staked tokens, and we select the top \( M \) candidates for the validator set, i.e., the threshold staked tokens to become a validator is the stake of the \( M \)th validator candidate, represented by \( s_{M} \). Accordingly, the rewards for a validator candidate \( n_k \) with \( s_k \) staked tokens for an epoch is as follows: 
\begin{equation}
r_{n_k}=\begin{cases}
          \alpha w_k \quad &\text{if} \, s_k > s_M,  \\
          0 \quad &\text{if} \, s_k \leq s_M. \\
     \end{cases}
\end{equation}

Capping the validator set cardinality is justified for two reasons. Firstly, in line with classical consensus mechanisms, an increase in the number of validators tends to decrease the system scalability, measured in throughput and latency~\cite{yin2019hotstuff}. Secondly, by imposing a maximum limit on the number of validators, and a minimum capital requirement of \( s_{M} \) staked tokens for running a validator, the system discourages single entities from dominating the validator set (i.e., Sybil attacks), further explored in the following section.

This method, implemented in blockchains such as Axelar~\cite{axelarBlockchain} and Celo~\cite{celoBlockchain}, helps in balancing scalability and decentralization.

\subsection{SRSW Function}
Building on these primitives, we now focus on the challenge of achieving equitable influence among validators to improve the Nakamoto and Gini coefficients for a given validator set.

We propose the \textit{Square Root Stake Weight (SRSW)} function, a novel approach that diverges from traditional linear weightings in quorum \( \mathbb{Q}' \) computations. The SRSW function calculates the weight \( w^*_i \) of each validator \( n_i \) based on the square root of their staked tokens \( s_i \), as defined by:
\vspace{-4pt}
\begin{equation}
   w^*_i = \sqrt{s_i} 
   \vspace{-4pt}
\end{equation}

The revised quorum \( \mathbb{Q}^* \) for the validator set \( N \) is formulated as follows:
\vspace{-6pt}
\begin{equation}
\mathbb{Q}^* \geq \left( \frac{2}{3}\right) \sum_{i} w_i^* = \left( \frac{2}{3}\right) \sum_{i} \sqrt{s_i} \quad \forall i \in N
  \vspace{-4pt}
\end{equation}

Contrasting with linear models, the SRSW function aims to reduce the disproportionate influence of validators with high staked tokens. In essence, the SRSW approach diminishes the weight disparities between validators with varying stake amounts. 

The validator rewards \( r^*_{n_i} \) are structured to reinforce rational decisions. The reward formula is:
\begin{equation}
    r^*_{n_i} = \alpha w_i^* = \alpha \sqrt{s_i}, \text{ if } s_i > s_M; \text{ otherwise, } 0.
\end{equation}
This incentivizes validators to keep or increase their stakes above the threshold \( s_M \), aligning individual gains with the system's stability. In other words, the system should satisfy the following condition. 
\vspace{-4pt}
\begin{equation}
    r_{n_i}^* > r_{n_j}^* + r_{n_k}^* - C, \text{ where } s_i \geq s_j + s_k
    \vspace{-4pt}
\end{equation}
This inequality implies that a validator with a combined stake \( s_i \) gains more rewards by maintaining a single identity rather than dividing into multiple validators with smaller stakes \( s_j \) and \( s_k \).

Consider a validator with \( s_i = 4 \) and \( s_M = 3 \). The options are: split into \( s_j^1 = 2, s_k^1 = 2 \), or \( s_j^2 = 3, s_k^2 = 1 \), or not split. In the first case, \( s_j^1, s_k^1 < s_M \) yield no rewards. In the second, rewards are \( \alpha \sqrt{3} \) for \( s_j^2 \) only, as \( s_k^2 < s_M \). Not splitting, \( \alpha \sqrt{4} \), offers the highest reward. Our approach effectively deters stake fragmentation and mitigates Sybil attacks, promoting consolidated stakes as a strategically rational choice.

When both \( s_i \) and \( s_j \) exceed \( s_M \), the inequality is adjusted in terms of weights:
\vspace{-2pt}
\begin{equation}
    \sqrt{s_i} > \sqrt{s_j} + \sqrt{s_k} - \frac{C}{\alpha},\text{ where } s_i \geq s_j + s_k \text{ and } s_i, s_j > s_M
\end{equation}

Here, our assumption of high \( C \) plays a crucial role to make the division of validator stakes non-rational, which we explore in the subsequent section.

\subsection{Discussion}

\textbf{Determining M.}
A critical aspect of this approach is the determination of \( M \), the maximum cardinality of the validator set. A low \( M \) might risk insufficient decentralization, while an excessively high \( M \) could impact system performance due to increased communication complexity. Although a fixed \( M \) may appear counterintuitive to decentralization, delegation in DPoS mechanisms enable individual token holders to collectively participate in consensus, thereby mitigating potential centralization concerns~\cite{saad2020comparative}. In this work, we do not prescribe a specific value for \( M \), as it depends on the algorithm and implementation. In practice, we have observed around one hundred validators is the ideal number for current algorithms~\cite{yin2019hotstuff,0LNetwork}.

\textbf{Sybil costs.} We assume that \( C \) is high. In practice, this remains an open challenge, particularly in token-based systems~\cite{messias2023airdrops,messias2023understanding}. Potential solutions include the use of proof of personhood~\cite{worldcoin,borge2017proof}, limiting one validator per geospatial location~\cite{motepalli2023analyzing}, KYC compliance~\cite{hodgson2002know}, or a combination of these. While these approaches might mitigate the problem to some degree, they may come at the cost of lost anonymity. Furthermore, detecting cartels among validators is challenging, especially at the protocol layer; therefore, we do not address this issue in the consensus mechanism. In conclusion, we acknowledge that establishing Sybil costs is more a complex socio-economic challenge than a technical one and is beyond the scope of this work.

In summary, we propose the SRSW function with corresponding \( \mathbb{Q}' \) and \( r^* \), and provide considerations on \( M \) and \( C \). We now turn our attention to evaluating this approach.

\section{Evaluation: Improved Decentralization}
\label{sec:evaluation}
In this section, we demonstrate that the SRSW function achieves higher decentralization, as measured by \( G, \rho_{\mathbb{N}_L}, \) and \( \rho_{\mathbb{N}_S} \), compared to the linear model. We then reinforce these claims with empirical evidence.

\subsection{Decentralization Metrics Analysis}
Let us analyze how the SRSW model offers a better Nakamoto coefficient than the linear model. 
\begin{theorem}
\label{theorem:liveness}
Given a validator set \( N \), the SRSW model's Nakamoto coefficient for liveness, \( \rho_{\mathbb{N}_L}^* \), is greater than or equal to that of the linear stake-weight model, \( \rho_{\mathbb{N}_L} \).
\end{theorem}
\vspace{-10pt}
\begin{proof}
Consider the Nakamoto coefficient computation in linear stake weight, let $K$ be the smallest subset such that: 
\vspace{-2pt}
\begin{equation}
    \sum_{n_i\in K} s_i > \frac{1}{3} \sum_{n_i\in N} s_i
\end{equation}
Similary in the SRSW model, let $K^*$ be the smallest subset satisfying the condition:
\vspace{-4pt}
\begin{equation}
    \sum_{i \in K^*} \sqrt{s_i} > \frac{1}{3} \sum_{i\in N} \sqrt{s_i}
\end{equation}
The concave nature of the square root function, as per Jensen's inequality~\cite{abramovich2004refining}, necessitates a larger $K^*$ to fulfill this condition in the SRSW model compared to $K$ in the linear model, thereby implying:
\vspace{-6pt}
\begin{equation}
    |K^*| \geq |K| \implies N_L^* \geq N_L
\end{equation}
\begin{equation}
    \frac{N_L^*}{m} \geq  \frac{N_L}{m} \implies \rho_{\mathbb{N}_L}^* > \rho_{\mathbb{N}_L}
\end{equation}
Hence, the relative Nakamoto coefficient for liveness is higher for SRSW compared to linear model. 
\end{proof}

Similarly, we can prove for Nakamoto coefficient-safety. 
\begin{theorem}
\label{theorem:safety}
Given \( N \), the SRSW model's Nakamoto coefficient for safety is greater than or equal to that of the linear stake-weight model.
\end{theorem}

\begin{theorem}
Given \( N \), the Gini indices of the SRSW and linear models, represented by \( G^* \) and \( G \), respectively, satisfy \( G^* \leq G \).
\end{theorem}

\begin{proof}
This proof draws upon the established principles from Theorem~\ref{theorem:liveness} and the definitions provided in Section~\ref{sec:metrics}.

In the SRSW model, the square root transformation applied to validator stakes results in a more uniform distribution of weights, effectively reducing relative disparities in stake sizes compared to the linear model. Consequently, this leads to a lower Gini index in the SRSW model, compared to the Gini index in the linear model.

Therefore, \( G^* \leq G \), indicating a more equitable distribution of validator influence under the SRSW model.
\end{proof}
\vspace{-5pt}

\begin{figure*}[htbp]
    \vspace{-12pt}
  \centering
  \includegraphics[trim={1cm 0.8cm 1cm 2cm}, width=\textwidth]{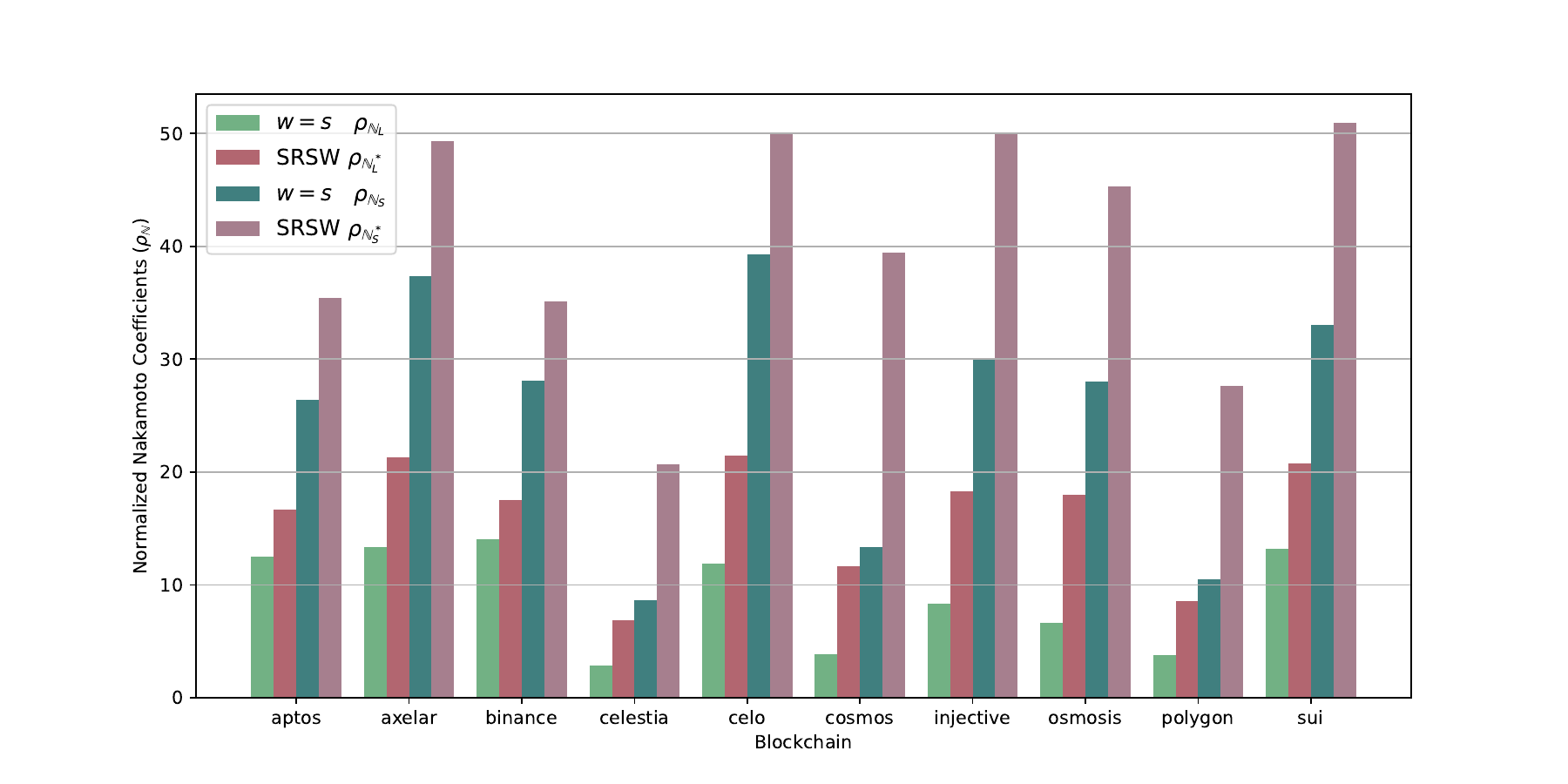}
  \caption{Comparison of Nakamoto coefficients for safety and liveness in linear and SRSW weighting functions}
  \label{fig:nakamoto_coefficients}
  \vspace{-14pt}
\end{figure*}

\begin{table}[]
\centering
\caption{Percentage decrease in Gini and percentage increase in Nakamoto coefficients with SRSW}
\label{tab:percentimprovement}
\vspace{-5pt}
\renewcommand{\arraystretch}{1.5}
\begin{tabular}{|l|l|l|l|}
\hline
\textbf{} & \textbf{\( G \) \% decrease}   &\textbf{ \( \mathbb{N}_L \%\) increase }   & \textbf{\( \mathbb{N}_S \%\) increase}    \\ \hline \hline
Aptos     & 26.78 & 33.33  & 34.21  \\ \hline
Axelar    & 39.02    & 60.0 & 32.14  \\ \hline
Binance   & 25.45 & 25    & 25     \\ \hline
Celestia &  22.89 & 140   & 140       \\ \hline
Celo      & 35 & 80    & 27.27  \\ \hline
Cosmos    & 45.58 & 200   & 195.83    \\ \hline
Injective & 48.97 & 120   & 66.66    \\ \hline
Osmosis   & 46.29 & 170   & 61.90  \\ \hline
Polygon   & 32.05 & 125   & 163.63 \\ \hline
Sui       & 48.78 & 57.14 & 54.28 \\ \hline \hline
\textbf{mean} & \textbf{37.16 }&\textbf{ 101.04} & \textbf{80.09} \\ \hline
\end{tabular}
\vspace{-18pt}
\end{table}

\begin{figure*}[ht]
    \centering
    \begin{subfigure}[t]{0.32\textwidth}
        \centering
        \includegraphics[width=\linewidth, trim= {0.5cm 0.5cm 0.5cm 0.5cm}]{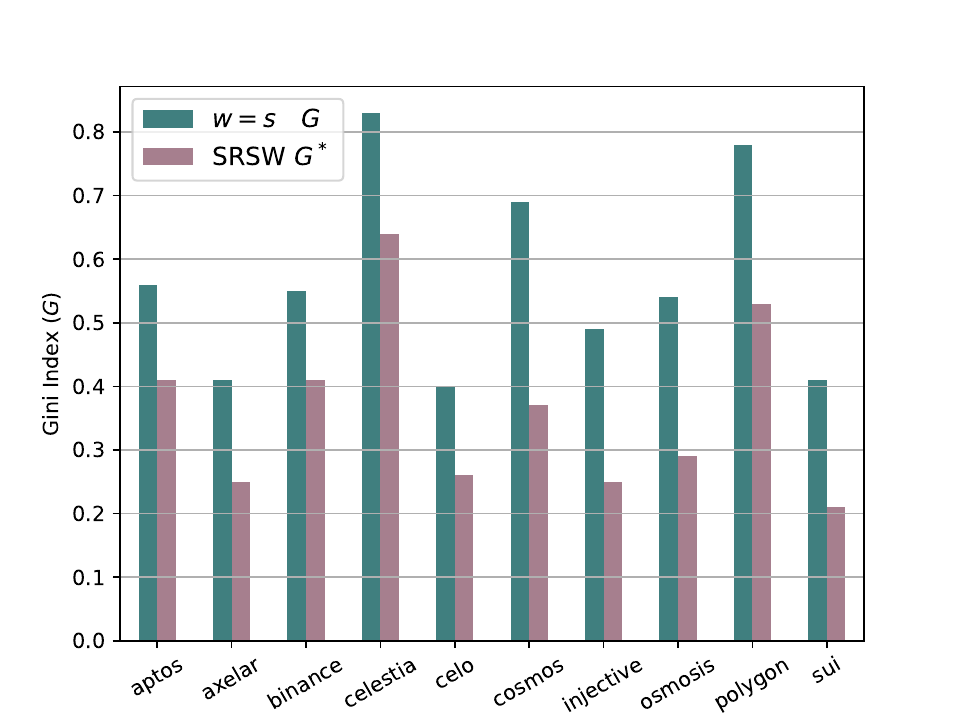}
        \caption{Comparison of Gini index}
        \label{fig:gini}
    \end{subfigure}
    \begin{subfigure}[t]{0.32\textwidth}
        \centering
        \includegraphics[width=\linewidth, trim= {3cm 0.5cm 2.8cm 250cm}]{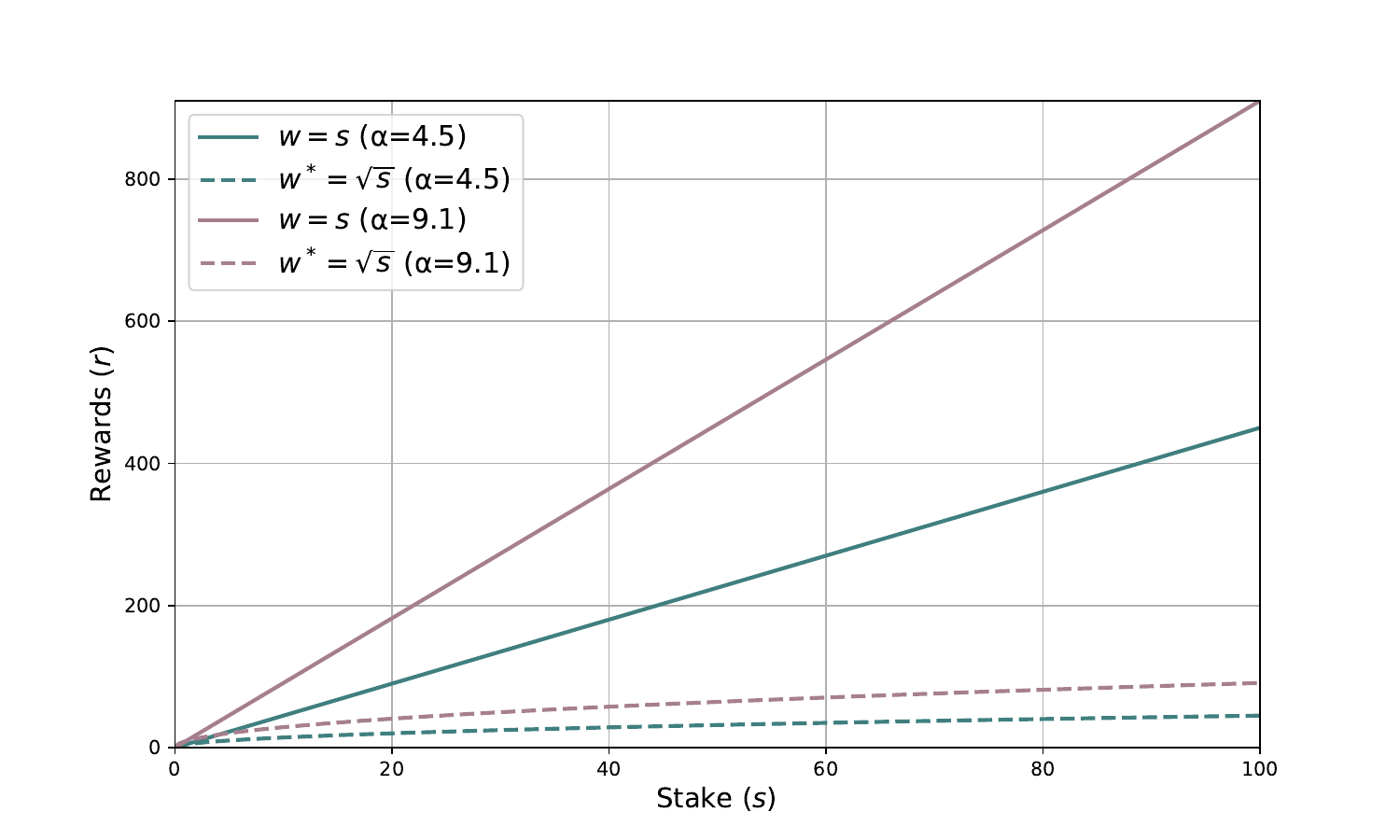}
        \caption{Reward rate with varying inflations}
        \label{fig:reward}
    \end{subfigure}
    \begin{subfigure}[t]{0.32\textwidth}
        \centering
        \includegraphics[width=\linewidth, trim= {0.5cm 0.5cm 0.5cm 0.5cm}]{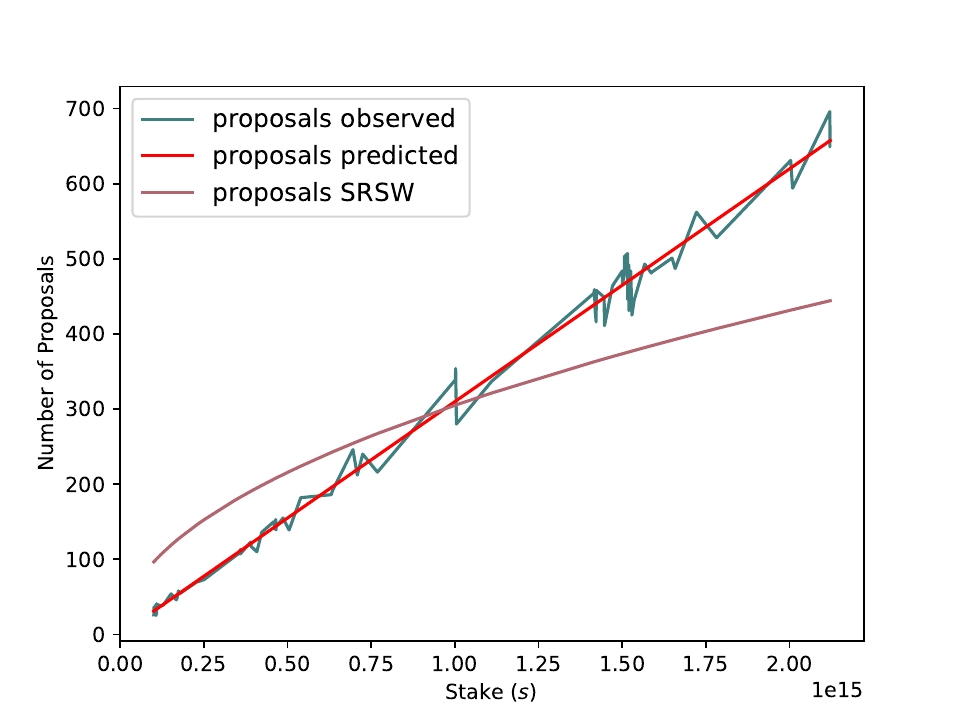}
        \caption{Block proposals made by validators}
        \label{fig:proposals}
    \end{subfigure}
    \caption{Evaluation of SRSW against linear stake weight}
    \label{fig:three_images}
    \vspace{-20pt}
\end{figure*}

\subsection{Empirical Validation}
In this subsection, we use the validator set data outlined in Section~\ref{sec:empiricalanalysis} to recalculate weights as per SRSW function. We then compare decentralization metrics of the SRSW model against the linear stake weight model for specified validator sets.

\subsubsection{Key Findings}
\begin{itemize}
    \item \textbf{Improvement in Gini Index:} With SRSW function, \( G \) observed a decrease ranging from 22.89\% to 48.97\%, as shown in Figure~\ref{fig:gini} and Table~\ref{tab:percentimprovement}. This reduction indicates a more equitable distribution of weight, enhancing the decentralization.
    \item \textbf{Increase in Nakamoto Coefficients:} As observed in Figure~\ref{fig:nakamoto_coefficients}, the Nakamoto coefficients for SRSW are superior to the current approach. In terms of absolute numbers, as shown in Table~\ref{tab:percentimprovement}, \(\mathbb{N}_L \) saw an increase between 25\% and 200\%, while \(\mathbb{N}_S \) had the increment range from 25\% to 195.83\%. These increases imply a stronger resistance to centralization by SRSW, requiring higher number of validators to influence the consensus.

\end{itemize}
\subsubsection{Implications}



\begin{itemize}
    \item \textbf{Rewards - Rate of Growth:} 
    In Figure~\ref{fig:reward}, we analyze the growth rate of rewards using two annual inflation rates ($\alpha = {4.5,9.1}$), chosen based on typical values observed in current blockchain implementations~\cite{stakingrewards}. This analysis underscores the benefits of the SRSW model, particularly in moderating the reward growth for validators with larger stakes. By addressing the 'rich get richer' narrative, the SRSW model promotes a fairer reward distribution, leading to more equitable weight compounding across epochs.

    \item \textbf{Block Generation Decentralization:} 
    Figure~\ref{fig:proposals} explores the dynamics of block proposal generation. Utilizing data from the Aptos blockchain~\cite{aptosBlockchain}, we initially demonstrate how the current block proposers are chosen based on their linear stake weights. Subsequently, we compare this against the predicted block proposer distribution under the SRSW model. Our findings reveal that the SRSW model leads to a more diverse array of block proposers, playing a crucial role in mitigating Miner Extractable Value (MEV) risks~\cite{daian2020flash} and enhancing censorship resistance~\cite{censorshipData}. This diversification in proposers is consistent with decentralization metrics discussed in contemporary studies~\cite{li2020comparison,li2023cross}, highlighting the SRSW model’s contribution to decentralization.
\end{itemize}
\vspace{-8pt}
\section{Related Work}
\label{sec:relatedwork}
Decentralization in blockchains, a cornerstone for blockchain efficacy, have been extensively explored, with a focus on governance~\cite{balajidecentralization, fritsch2022analyzing, kiayias2022sok, sharma2023unpacking, tan2023open}. Recent studies, such as those on token-based voting in DeFi protocols~\cite{messias2023understanding}, underscore the evolving complexities in blockchain governance. Particularly, decentralization research in PoS and DPoS systems have illuminated the challenges of weight concentration in these protocols~\cite{lin2021measuring, li2023liquid, liu2022understanding, li2023cross, li2020comparison, li2023hard, liu2022decentralization, kim2019stellar}.

These studies identify and quantify the weight concentration challenge in weighted consensus, yet solutions to this issue remain underexplored. Existing suggestions, such as capping proposals per validator~\cite{jeong2020centralized} and introducing reward sharing in validator staking pools~\cite{brunjes2020reward}, offer only partial remedies. Latest research on introducing virtual stake based on validator performance~\cite{mivsic2023towards} is addressing the challenge but raises potential vulnerabilities such as the `nothing-at-stake' problem~\cite{motepalli2021reward}. Our work diverges by enhancing decentralization directly at the consensus mechanism, without the need for new tokens or introducing associated vulnerabilities, thus improving decentralization more holistically.
\vspace{-6pt}
\section{Conclusions}
\label{sec:conclusion}
In this study, we introduced the Square Root Stake Weight (SRSW) function to address weight concentration in permissionless blockchains, demonstrating substantial improvements in decentralization metrics such as the Gini index and Nakamoto coefficients. While acknowledging the intricacies related to Sybil cost augmentation, this paper highlights the necessity for further investigation into practical implementations and governance models. Future research directions could include exploring geospatial weight distribution and even auction-based mechanisms for validator set selection, potentially offering a means to further decentralize and economically optimize blockchain consensus mechanisms.
\section*{Acknowledgment}
This work has in part been supported by Sui Foundation and NSERC.
\bibliographystyle{plain}
\bibliography{references}

\begin{thebibliography}{10}

\bibitem{coinmarketcap}
\url{https://coinmarketcap.com/}.
\newblock Accessed: 2023-12-14.

\bibitem{0LNetwork}
0L.
\newblock Proof-of-fee, part 2.
\newblock \url{https://0l.network/2022/10/20/proof-of-fee-part-2-a-proposal/}.
\newblock Accessed: 2023-12-15.

\bibitem{abramovich2004refining}
Shoshana Abramovich, Graham Jameson, and Gord Sinnamon.
\newblock Refining jensen's inequality.
\newblock {\em Bulletin math{\'e}matique de la Soci{\'e}t{\'e} des Sciences Math{\'e}matiques de Roumanie}, pages 3--14, 2004.

\bibitem{injectiveBlog}
Big Ace.
\newblock Injective tendermint core: A powerful consensus engine for decentralized finance.
\newblock \url{https://medium.com/@charlesace/injective-tendermint-core-a-powerful-consensus-engine-for-decentralized-finance-a1db298b0b70}.
\newblock Accessed: 2023-12-14.

\bibitem{al2019lazyledger}
Mustafa Al-Bassam.
\newblock Lazyledger: A distributed data availability ledger with client-side smart contracts.
\newblock {\em arXiv preprint arXiv:1905.09274}, 2019.

\bibitem{anceaume2020finality}
Emmanuelle Anceaume, Antonella Pozzo, Thibault Rieutord, and Sara Tucci-Piergiovanni.
\newblock On finality in blockchains.
\newblock {\em arXiv preprint arXiv:2012.10172}, 2020.

\bibitem{aptos}
Aptos.
\newblock \url{https://aptosfoundation.org/}.
\newblock Accessed: 2023-10-25.

\bibitem{austgen2023dao}
James Austgen, Andr{\'e}s F{\'a}brega, Sarah Allen, Kushal Babel, Mahimna Kelkar, and Ari Juels.
\newblock Dao decentralization: Voting-bloc entropy, bribery, and dark daos.
\newblock {\em arXiv preprint arXiv:2311.03530}, 2023.

\bibitem{axelar}
Axelar.
\newblock \url{https://axelar.network/}.
\newblock Accessed: 2023-10-25.

\bibitem{baudet2019state}
Mathieu Baudet, Avery Ching, Andrey Chursin, George Danezis, Fran{\c{c}}ois Garillot, Zekun Li, Dahlia Malkhi, Oded Naor, Dmitri Perelman, and Alberto Sonnino.
\newblock State machine replication in the libra blockchain.
\newblock {\em The Libra Assn., Tech. Rep}, 7, 2019.

\bibitem{belchior2021survey}
Rafael Belchior, Andr{\'e} Vasconcelos, S{\'e}rgio Guerreiro, and Miguel Correia.
\newblock A survey on blockchain interoperability: Past, present, and future trends.
\newblock {\em ACM Computing Surveys (CSUR)}, 54(8):1--41, 2021.

\bibitem{borge2017proof}
Maria Borge, Eleftherios Kokoris-Kogias, Philipp Jovanovic, Linus Gasser, Nicolas Gailly, and Bryan Ford.
\newblock Proof-of-personhood: Redemocratizing permissionless cryptocurrencies.
\newblock In {\em 2017 IEEE European Symposium on Security and Privacy Workshops (EuroS\&PW)}, pages 23--26. IEEE, 2017.

\bibitem{brunjes2020reward}
Lars Br{\"u}njes, Aggelos Kiayias, Elias Koutsoupias, and Aikaterini-Panagiota Stouka.
\newblock Reward sharing schemes for stake pools.
\newblock In {\em 2020 IEEE european symposium on security and privacy (EuroS\&p)}, pages 256--275. IEEE, 2020.

\bibitem{buchman2016tendermint}
Ethan Buchman.
\newblock {\em Tendermint: Byzantine fault tolerance in the age of blockchains}.
\newblock PhD thesis, University of Guelph, 2016.

\bibitem{buchman2018latest}
Ethan Buchman, Jae Kwon, and Zarko Milosevic.
\newblock The latest gossip on bft consensus.
\newblock {\em arXiv preprint arXiv:1807.04938}, 2018.

\bibitem{castro1999practical}
Miguel Castro, Barbara Liskov, et~al.
\newblock Practical byzantine fault tolerance.
\newblock In {\em OsDI}, volume~99, pages 173--186, 1999.

\bibitem{celestia}
Celestia.
\newblock \url{https://celo.org/}.
\newblock Accessed: 2023-12-14.

\bibitem{celo}
Celo.
\newblock \url{https://celo.org/}.
\newblock Accessed: 2023-12-14.

\bibitem{ceriani2012origins}
Lidia Ceriani and Paolo Verme.
\newblock The origins of the gini index: extracts from variabilit{\`a} e mutabilit{\`a} (1912) by corrado gini.
\newblock {\em The Journal of Economic Inequality}, 10:421--443, 2012.

\bibitem{bnb}
BNB Chain.
\newblock \url{https://www.bnbchain.org/}.
\newblock Accessed: 2023-10-25.

\bibitem{bnbBlockchain}
BNB~Smart Chain.
\newblock White paper.
\newblock \url{https://github.com/bnb-chain/whitepaper/blob/master/WHITEPAPER.md}.
\newblock Accessed: 2023-11-08.

\bibitem{cosmos}
Cosmos.
\newblock \url{https://cosmos.network/}.
\newblock Accessed: 2023-10-25.

\bibitem{daian2020flash}
Philip Daian, Steven Goldfeder, Tyler Kell, Yunqi Li, Xueyuan Zhao, Iddo Bentov, Lorenz Breidenbach, and Ari Juels.
\newblock Flash boys 2.0: Frontrunning in decentralized exchanges, miner extractable value, and consensus instability.
\newblock In {\em 2020 IEEE Symposium on Security and Privacy (SP)}, pages 910--927. IEEE, 2020.

\bibitem{david2018ouroboros}
Bernardo David, Peter Ga{\v{z}}i, Aggelos Kiayias, and Alexander Russell.
\newblock Ouroboros praos: An adaptively-secure, semi-synchronous proof-of-stake blockchain.
\newblock In {\em Advances in Cryptology--EUROCRYPT 2018: 37th Annual International Conference on the Theory and Applications of Cryptographic Techniques, Tel Aviv, Israel, April 29-May 3, 2018 Proceedings, Part II 37}, pages 66--98. Springer, 2018.

\bibitem{de2018pbft}
Stefano De~Angelis, Leonardo Aniello, Roberto Baldoni, Federico Lombardi, Andrea Margheri, Vladimiro Sassone, et~al.
\newblock Pbft vs proof-of-authority: Applying the cap theorem to permissioned blockchain.
\newblock In {\em CEUR workshop proceedings}, volume 2058. CEUR-WS, 2018.

\bibitem{aptosBlockchain}
Aptos Dev.
\newblock The aptos blockchain: Safe, scalable, and upgradeable web3 infrastructure.
\newblock \url{https://aptos.dev/aptos-white-paper/}.
\newblock Published:2022-08-11, v1.0.

\bibitem{celestiaDocs}
Celestia Docs.
\newblock Celestia's data availability layer.
\newblock \url{https://docs.celestia.org/learn/how-celestia-works/data-availability-layer}.
\newblock Accessed: 2023-12-14.

\bibitem{celoBlockchain}
Celo docs.
\newblock Consensus.
\newblock \url{https://docs.celo.org/protocol/consensus}.
\newblock Accessed: 2023-11-26.

\bibitem{osmosisBlockchain}
Osmosis Docs.
\newblock Glossary.
\newblock \url{https://docs.osmosis.zone/overview/educate/terminology#consensus}.
\newblock Accessed: 2023-11-26.

\bibitem{suiBlockchain}
Sui Docs.
\newblock Validator committee.
\newblock \url{https://docs.sui.io/guides/operator/validator-committee}.
\newblock Accessed: 2023-11-26.

\bibitem{douceur2002sybil}
John~R Douceur.
\newblock The sybil attack.
\newblock In {\em International workshop on peer-to-peer systems}, pages 251--260. Springer, 2002.

\bibitem{ofacSanctioned}
Justin Drake and Toni Wahrstätter.
\newblock Ethereum addresses added to the ofac sdn list.
\newblock \url{https://github.com/ultrasoundmoney/ofac-ethereum-addresses/}.
\newblock Accessed: 2023-12-17.

\bibitem{duan2022foundations}
Sisi Duan and Haibin Zhang.
\newblock Foundations of dynamic bft.
\newblock In {\em 2022 IEEE Symposium on Security and Privacy (SP)}, pages 1317--1334. IEEE, 2022.

\bibitem{fritsch2022analyzing}
Robin Fritsch, Marino M{\"u}ller, and Roger Wattenhofer.
\newblock Analyzing voting power in decentralized governance: Who controls daos?
\newblock {\em arXiv preprint arXiv:2204.01176}, 2022.

\bibitem{garay2015bitcoin}
Juan Garay, Aggelos Kiayias, and Nikos Leonardos.
\newblock The bitcoin backbone protocol: Analysis and applications.
\newblock In {\em Annual international conference on the theory and applications of cryptographic techniques}, pages 281--310. Springer, 2015.

\bibitem{gastwirth1972estimation}
Joseph~L Gastwirth.
\newblock The estimation of the lorenz curve and gini index.
\newblock {\em The review of economics and statistics}, pages 306--316, 1972.

\bibitem{gelashvili2022jolteon}
Rati Gelashvili, Lefteris Kokoris-Kogias, Alberto Sonnino, Alexander Spiegelman, and Zhuolun Xiang.
\newblock Jolteon and ditto: Network-adaptive efficient consensus with asynchronous fallback.
\newblock In {\em International conference on financial cryptography and data security}, pages 296--315. Springer, 2022.

\bibitem{gilad2017algorand}
Yossi Gilad, Rotem Hemo, Silvio Micali, Georgios Vlachos, and Nickolai Zeldovich.
\newblock Algorand: Scaling byzantine agreements for cryptocurrencies.
\newblock In {\em Proceedings of the 26th symposium on operating systems principles}, pages 51--68, 2017.

\bibitem{gini1921measurement}
Corrado Gini.
\newblock Measurement of inequality of incomes.
\newblock {\em The economic journal}, 31(121):124--125, 1921.

\bibitem{axelarBlockchain}
GitHub.
\newblock Axelar core.
\newblock \url{https://github.com/axelarnetwork/axelar-core/blob/main/docs/cli/axelard_tendermint_version.md}.
\newblock Accessed: 2023-11-26.

\bibitem{gueta2019sbft}
Guy~Golan Gueta, Ittai Abraham, Shelly Grossman, Dahlia Malkhi, Benny Pinkas, Michael Reiter, Dragos-Adrian Seredinschi, Orr Tamir, and Alin Tomescu.
\newblock Sbft: A scalable and decentralized trust infrastructure.
\newblock In {\em 2019 49th Annual IEEE/IFIP international conference on dependable systems and networks (DSN)}, pages 568--580. IEEE, 2019.

\bibitem{hodgson2002know}
Damian Hodgson.
\newblock “know your customer”: marketing, governmentality and the “new consumer” of financial services.
\newblock {\em Management Decision}, 40(4):318--328, 2002.

\bibitem{injective}
Injective.
\newblock \url{https://injective.com/}.
\newblock Accessed: 2023-12-14.

\bibitem{jeong2020centralized}
Seungwon~Eugene Jeong.
\newblock Centralized decentralization: Does voting matter? simple economics of the dpos blockchain governance.
\newblock {\em Simple Economics of the DPoS Blockchain Governance (April 21, 2020)}, 2020.

\bibitem{kiayias2022sok}
Aggelos Kiayias and Philip Lazos.
\newblock Sok: blockchain governance.
\newblock {\em arXiv preprint arXiv:2201.07188}, 2022.

\bibitem{ethereumCoindesk}
Christine Kim.
\newblock The most pressing issue on ethereum is validator size growth.
\newblock \url{https://www.coindesk.com/consensus-magazine/2023/09/29/the-most-pressing-issue-on-ethereum-is-validator-size-growth/}.
\newblock Accessed: 2023-12-14.

\bibitem{kim2023taxonomic}
Heesang Kim and Dohoon Kim.
\newblock A taxonomic hierarchy of blockchain consensus algorithms: An evolutionary phylogeny approach.
\newblock {\em Sensors}, 23(5):2739, 2023.

\bibitem{kim2019stellar}
Minjeong Kim, Yujin Kwon, and Yongdae Kim.
\newblock Is stellar as secure as you think?
\newblock In {\em 2019 IEEE European Symposium on Security and Privacy Workshops (EuroS\&PW)}, pages 377--385. IEEE, 2019.

\bibitem{kwon2019impossibility}
Yujin Kwon, Jian Liu, Minjeong Kim, Dawn Song, and Yongdae Kim.
\newblock Impossibility of full decentralization in permissionless blockchains.
\newblock In {\em Proceedings of the 1st ACM Conference on Advances in Financial Technologies}, pages 110--123, 2019.

\bibitem{lamport2019byzantine}
Leslie Lamport, Robert Shostak, and Marshall Pease.
\newblock The byzantine generals problem.
\newblock In {\em Concurrency: the works of leslie lamport}, pages 203--226. 2019.

\bibitem{lehar2021decentralized}
Alfred Lehar and Christine~A Parlour.
\newblock Decentralized exchanges.
\newblock {\em Available at SSRN 3905316}, 2021.

\bibitem{lewis2021does}
Andrew Lewis-Pye and Tim Roughgarden.
\newblock How does blockchain security dictate blockchain implementation?
\newblock In {\em Proceedings of the 2021 ACM SIGSAC Conference on Computer and Communications Security}, pages 1006--1019, 2021.

\bibitem{li2020comparison}
Chao Li and Balaji Palanisamy.
\newblock Comparison of decentralization in dpos and pow blockchains.
\newblock In {\em Blockchain--ICBC 2020: Third International Conference, Held as Part of the Services Conference Federation, SCF 2020, Honolulu, HI, USA, September 18-20, 2020, Proceedings 3}, pages 18--32. Springer, 2020.

\bibitem{li2023cross}
Chao Li, Balaji Palanisamy, Runhua Xu, and Li~Duan.
\newblock Cross-consensus measurement of individual-level decentralization in blockchains.
\newblock In {\em 2023 IEEE 9th Intl Conference on Big Data Security on Cloud (BigDataSecurity), IEEE Intl Conference on High Performance and Smart Computing,(HPSC) and IEEE Intl Conference on Intelligent Data and Security (IDS)}, pages 45--50. IEEE, 2023.

\bibitem{li2023hard}
Chao Li, Balaji Palanisamy, Runhua Xu, Li~Duan, Jiqiang Liu, and Wei Wang.
\newblock How hard is takeover in dpos blockchains? understanding the security of coin-based voting governance.
\newblock {\em arXiv preprint arXiv:2310.18596}, 2023.

\bibitem{li2023liquid}
Chao Li, Runhua Xu, and Li~Duan.
\newblock Liquid democracy in dpos blockchains.
\newblock In {\em Proceedings of the 5th ACM International Symposium on Blockchain and Secure Critical Infrastructure}, pages 25--33, 2023.

\bibitem{lin2021measuring}
Qinwei Lin, Chao Li, Xifeng Zhao, and Xianhai Chen.
\newblock Measuring decentralization in bitcoin and ethereum using multiple metrics and granularities.
\newblock In {\em 2021 IEEE 37th International Conference on Data Engineering Workshops (ICDEW)}, pages 80--87. IEEE, 2021.

\bibitem{liu2022decentralization}
Jieli Liu, Weilin Zheng, Dingyuan Lu, Jiajing Wu, and Zibin Zheng.
\newblock From decentralization to oligopoly: A data-driven analysis of decentralization evolution and voting behaviors on eosio.
\newblock {\em IEEE Transactions on Computational Social Systems}, 2022.

\bibitem{liu2022understanding}
Jieli Liu, Weilin Zheng, Dingyuan Lu, Jiajing Wu, and Zibin Zheng.
\newblock Understanding the decentralization of dpos: perspectives from data-driven analysis on eosio.
\newblock {\em arXiv preprint arXiv:2201.06187}, 2022.

\bibitem{mccorry2021sok}
Patrick McCorry, Chris Buckland, Bennet Yee, and Dawn Song.
\newblock Sok: Validating bridges as a scaling solution for blockchains.
\newblock {\em Cryptology ePrint Archive}, 2021.

\bibitem{messias2023understanding}
Johnnatan Messias, Vabuk Pahari, Balakrishnan Chandrasekaran, Krishna~P Gummadi, and Patrick Loiseau.
\newblock Understanding blockchain governance: Analyzing decentralized voting to amend defi smart contracts.
\newblock {\em arXiv preprint arXiv:2305.17655}, 2023.

\bibitem{messias2023airdrops}
Johnnatan Messias, Aviv Yaish, and Benjamin Livshits.
\newblock Airdrops: Giving money away is harder than it seems.
\newblock {\em arXiv preprint arXiv:2312.02752}, 2023.

\bibitem{miller2016honey}
Andrew Miller, Yu~Xia, Kyle Croman, Elaine Shi, and Dawn Song.
\newblock The honey badger of bft protocols.
\newblock In {\em Proceedings of the 2016 ACM SIGSAC conference on computer and communications security}, pages 31--42, 2016.

\bibitem{mivsic2023towards}
Jelena Mi{\v{s}}i{\'c}, Vojislav~B Mi{\v{s}}i{\'c}, and Xiaolin Chang.
\newblock Towards decentralization in dpos systems: election, voting and leader selection using virtual stake.
\newblock {\em IEEE Transactions on Network and Service Management}, 2023.

\bibitem{moniz2020istanbul}
Henrique Moniz.
\newblock The istanbul bft consensus algorithm.
\newblock {\em arXiv preprint arXiv:2002.03613}, 2020.

\bibitem{lidoArticle}
Nicholas Morgan.
\newblock Lido dominance prompts warnings about liquid staking derivatives.
\newblock \url{https://decrypt.co/154804/lido-lsd-liquid-staking-decentralization}.
\newblock Accessed: 2023-11-29.

\bibitem{motepalli2023sok}
Shashank Motepalli, Luciano Freitas, and Benjamin Livshits.
\newblock Sok: Decentralized sequencers for rollups.
\newblock {\em arXiv preprint arXiv:2310.03616}, 2023.

\bibitem{motepalli2021reward}
Shashank Motepalli and Hans-Arno Jacobsen.
\newblock Reward mechanism for blockchains using evolutionary game theory.
\newblock In {\em 2021 3rd Conference on Blockchain Research \& Applications for Innovative Networks and Services (BRAINS)}, pages 217--224. IEEE, 2021.

\bibitem{motepalli2022decentralizing}
Shashank Motepalli and Hans-Arno Jacobsen.
\newblock Decentralizing permissioned blockchain with delay towers.
\newblock {\em arXiv preprint arXiv:2203.09714}, 2022.

\bibitem{motepalli2023analyzing}
Shashank Motepalli and Hans-Arno Jacobsen.
\newblock Analyzing geospatial distribution in blockchains.
\newblock {\em arXiv preprint arXiv:2305.17771}, 2023.

\bibitem{nakamoto2008bitcoin}
Satoshi Nakamoto.
\newblock Bitcoin: A peer-to-peer electronic cash system.
\newblock {\em Decentralized business review}, 2008.

\bibitem{ongaro2014search}
Diego Ongaro and John Ousterhout.
\newblock In search of an understandable consensus algorithm.
\newblock In {\em 2014 USENIX annual technical conference (USENIX ATC 14)}, pages 305--319, 2014.

\bibitem{osmosis}
Osmosis.
\newblock \url{https://osmosis.zone/}.
\newblock Accessed: 2023-10-25.

\bibitem{pass2017analysis}
Rafael Pass, Lior Seeman, and Abhi Shelat.
\newblock Analysis of the blockchain protocol in asynchronous networks.
\newblock In {\em Annual international conference on the theory and applications of cryptographic techniques}, pages 643--673. Springer, 2017.

\bibitem{polygon}
Polygon.
\newblock \url{https://polygon.technology/}.
\newblock Accessed: 2023-10-25.

\bibitem{stakingrewards}
Staking Rewards.
\newblock Proof of stake.
\newblock \url{https://www.stakingrewards.com/assets/proof-of-stake}.
\newblock Accessed: 2023-12-03.

\bibitem{saad2020comparative}
Sheikh Munir~Skh Saad and Raja Zahilah Raja~Mohd Radzi.
\newblock Comparative review of the blockchain consensus algorithm between proof of stake (pos) and delegated proof of stake (dpos).
\newblock {\em International Journal of Innovative Computing}, 10(2), 2020.

\bibitem{schneider2003decentralization}
Aaron Schneider.
\newblock Decentralization: Conceptualization and measurement.
\newblock {\em Studies in comparative international development}, 38:32--56, 2003.

\bibitem{sharma2023unpacking}
Tanusree Sharma, Yujin Kwon, Kornrapat Pongmala, Henry Wang, Andrew Miller, Dawn Song, and Yang Wang.
\newblock Unpacking how decentralized autonomous organizations (daos) work in practice.
\newblock {\em arXiv preprint arXiv:2304.09822}, 2023.

\bibitem{siewiorek2005fault}
Daniel~P Siewiorek and Priya Narasimhan.
\newblock Fault-tolerant architectures for space and avionics applications.
\newblock {\em NASA Ames Research http://ic. arc. nasa. gov/projects/ishem/Papers/Siewi}, 2005.

\bibitem{sitthiyot2020simple}
Thitithep Sitthiyot and Kanyarat Holasut.
\newblock A simple method for measuring inequality.
\newblock {\em Palgrave Communications}, 6(1):1--9, 2020.

\bibitem{sitthiyot2021simple}
Thitithep Sitthiyot and Kanyarat Holasut.
\newblock A simple method for estimating the lorenz curve.
\newblock {\em Humanities and Social Sciences Communications}, 8(1):1--9, 2021.

\bibitem{spiegelman2022bullshark}
Alexander Spiegelman, Neil Giridharan, Alberto Sonnino, and Lefteris Kokoris-Kogias.
\newblock Bullshark: Dag bft protocols made practical.
\newblock In {\em Proceedings of the 2022 ACM SIGSAC Conference on Computer and Communications Security}, pages 2705--2718, 2022.

\bibitem{sridhar2023better}
Srivatsan Sridhar, Dionysis Zindros, and David Tse.
\newblock Better safe than sorry: Recovering after adversarial majority.
\newblock {\em arXiv preprint arXiv:2310.06338}, 2023.

\bibitem{balajidecentralization}
Balaji~S. Srinivasan and Leland Lee.
\newblock Quantifying decentralization.
\newblock \url{https://news.earn.com/quantifying-decentralization-e39db233c28e}.
\newblock Accessed: 2023-11-05.

\bibitem{sui}
Sui.
\newblock \url{https://sui.io/}.
\newblock Accessed: 2023-10-25.

\bibitem{tan2023open}
Joshua~Z Tan, Tara Merk, Sarah Hubbard, Eliza~R Oak, Joni Pirovich, Ellie Rennie, Rolf Hoefer, Michael Zargham, Jason Potts, Chris Berg, et~al.
\newblock Open problems in daos.
\newblock {\em arXiv preprint arXiv:2310.19201}, 2023.

\bibitem{cosmosBlockchain}
Chjango Unchained.
\newblock Tendermint explained — bringing bft-based pos to the public blockchain domain.
\newblock \url{https://blog.cosmos.network/tendermint-explained-bringing-bft-based-pos-to-the-public-blockchain-domain-f22e274a0fdb}.
\newblock Accessed: 2023-11-26.

\bibitem{censorshipData}
Toni Wahrstätter.
\newblock Ethereum censorship dashboard.
\newblock \url{https://censorship.pics/}.
\newblock Accessed: 2023-12-17.

\bibitem{wang2021ethereum}
Zeli Wang, Hai Jin, Weiqi Dai, Kim-Kwang~Raymond Choo, and Deqing Zou.
\newblock Ethereum smart contract security research: survey and future research opportunities.
\newblock {\em Frontiers of Computer Science}, 15:1--18, 2021.

\bibitem{wensley1978sift}
John~H Wensley, Leslie Lamport, Jack Goldberg, Milton~W Green, Karl~N Levitt, Po~Mo Melliar-Smith, Robert~E Shostak, and Charles~B Weinstock.
\newblock Sift: Design and analysis of a fault-tolerant computer for aircraft control.
\newblock {\em Proceedings of the IEEE}, 66(10):1240--1255, 1978.

\bibitem{polygonBlockchain}
Polygon wiki.
\newblock Peppermint.
\newblock \url{https://wiki.polygon.technology/docs/pos/design/heimdall/peppermint}.
\newblock Accessed: 2023-11-26.

\bibitem{worldcoin}
Worldcoin.
\newblock A new identity and financial network.
\newblock \url{https://whitepaper.worldcoin.org/#a-new-identity-and-financial-network}.
\newblock Accessed: 2023-12-16.

\bibitem{yin2019hotstuff}
Maofan Yin, Dahlia Malkhi, Michael~K Reiter, Guy~Golan Gueta, and Ittai Abraham.
\newblock Hotstuff: Bft consensus with linearity and responsiveness.
\newblock In {\em Proceedings of the 2019 ACM Symposium on Principles of Distributed Computing}, pages 347--356, 2019.

\bibitem{edwardThesis}
Gengrui Zhang.
\newblock {\em Towards More Efficient and Scalable Consensus Algorithms}.
\newblock PhD thesis, University of Toronto, 2023.

\bibitem{zhang2022reaching}
Gengrui Zhang, Fei Pan, Michael Dang'ana, Yunhao Mao, Shashank Motepalli, Shiquan Zhang, and Hans-Arno Jacobsen.
\newblock Reaching consensus in the byzantine empire: A comprehensive review of bft consensus algorithms.
\newblock {\em arXiv preprint arXiv:2204.03181}, 2022.

\bibitem{zhang2023prestigebft}
Gengrui Zhang, Fei Pan, Sofia Tijanic, and Hans-Arno Jacobsen.
\newblock Prestigebft: Revolutionizing view changes in bft consensus algorithms with reputation mechanisms.
\newblock {\em arXiv preprint arXiv:2307.08154}, 2023.

\end{thebibliography}
\end{document}